\documentclass[11pt]{article}
\usepackage{fullpage}
\usepackage{amsfonts}
\usepackage{algorithmic}
\usepackage{algorithm}
\usepackage{amsmath,amsthm,amssymb}
\usepackage{latexsym}
\usepackage{epic}
\usepackage{epsfig}
\usepackage{amscd}
\usepackage{url}
\usepackage{verbatim}

\newtheorem{theorem}{Theorem}[section]
\newtheorem{corollary}[theorem]{Corollary}
\newtheorem{lemma}[theorem]{Lemma}
\newtheorem{proposition}[theorem]{Proposition}

\newtheorem{definition}[theorem]{Definition}

\newtheorem{observation}[theorem]{Observation}

\newenvironment{proofsketch}{\begin{proof}[Proof Sketch]}{\end{proof}}

\newcommand{\cross}{\times}

\newcommand{\set}[1]{\left\{ #1 \right\}}
\newcommand{\union}{\bigcup}
\newcommand{\intersect}{\bigcap}
\newcommand{\sm}{\setminus}

\renewcommand{\hat}{\widehat}

\def\max{\qopname\relax n{max}}

\def\argmax{\qopname\relax n{argmax}}

\def\Ex{\qopname\relax n{\mathbf{E}}}

\newcommand{\RR}{\mathbb{R}}
\newcommand{\RRp}{\RR^+}

\newcommand{\ZZ}{\mathbb{Z}}

\def\A{\mathcal{A}}

\def\V{\mathcal{V}}

\def\sse{\subseteq}

\newcommand{\eat}[1]{}

\newcommand{\figeps}[2][1.0]{
  \begin{center}
    \includegraphics[scale=#1]{#2}
  \end{center}}

\newcommand{\mini}[1]{\mbox{minimize} & {#1} &\\}
\newcommand{\maxi}[1]{\mbox{maximize} & {#1 } & \\}
\newcommand{\st}{\mbox{subject to} }
\newcommand{\con}[1]{&#1 & \\}
\newcommand{\qcon}[2]{&#1, & \mbox{for } #2.  \\}
\newenvironment{lp}{\begin{equation}  \begin{array}{lll}}{\end{array}\end{equation}}
\newenvironment{lp*}{\begin{equation*}  \begin{array}{lll}}{\end{array}\end{equation*}}

\newcommand{\ie}{{\it i.e.}}

\begin{document}

\title{Truthful Assignment without Money}



\author{
Shaddin Dughmi\thanks{Supported by a Yahoo! Research internship, a Siebel Foundation Scholarship, and NSF grant CCF-0448664.} \\
Stanford University\\
{\tt shaddin@cs.stanford.edu}
\and
Arpita Ghosh\\
Yahoo! Research\\
{\tt arpita@yahoo-inc.com}
}

\maketitle

\newcommand{\gap}{\textrm{GAP}}
\newcommand{\vigap}{\textrm{VIGAP}}
\newcommand{\sigap}{\textrm{SIGAP}}
\newcommand{\mkp}{\textrm{MKP}}
\newcommand{\kp}{\textrm{KP}}
\newcommand{\mwbm}{\textrm{MWBM}}
\newcommand{\mbm}{\textrm{MBM}}

\thispagestyle{empty}
\addtocounter{page}{-1}
\begin{abstract}
We study the design of truthful mechanisms that do not use payments for the generalized assignment problem (GAP) and its variants. An instance of the GAP consists of a bipartite graph with jobs on one side and machines on the other. Machines have capacities and edges have values and sizes; the goal is to construct a welfare maximizing feasible assignment. In our model of private valuations, motivated by impossibility results, the value and sizes on all job-machine pairs are public information; however, whether an edge {\em exists} or not in the bipartite graph is a job's private information. That is, the selfish agents in our model are the jobs, and their private information is their edge set. We want to design mechanisms that are truthful without money (henceforth \emph{strategyproof}), and produce assignments whose welfare is a good approximation to the optimal omniscient welfare.

We study several variants of the GAP starting with matching. For the unweighted version, we give an optimal strategyproof mechanism. For maximum weight bipartite matching, we show that no strategyproof mechanism, deterministic or randomized, can be optimal, and present a 2-approximate strategyproof mechanism along with a matching lowerbound. Next we study knapsack-like problems, which, unlike matching, are NP-hard. For these problems, we develop a general LP-based technique that extends the ideas of Lavi and Swamy \cite{laviswamy} to reduce designing a truthful approximate mechanism without money to designing such a mechanism for the {\em fractional} version of the problem.
We design strategyproof  approximate mechanisms for the fractional relaxations of multiple knapsack, size-invariant GAP, and value-invariant GAP, and use this technique to obtain, respectively, $2$, $4$ and $4$-approximate strategyproof mechanisms for these problems. We then design an $O(\log n)$-approximate strategyproof mechanism for the GAP by reducing, with logarithmic loss in the approximation, to our solution for the value-invariant GAP. Our technique may be of independent interest for designing truthful mechanisms without money for other LP-based problems.
\end{abstract}
\newpage

\section{Introduction}\label{sec:intro}
The design of truthful mechanisms, where selfish utility maximizing agents have no incentive to lie about their true preferences, has been studied in innumerable settings. The vast majority of these mechanisms, however, assume the existence of money--- a carefully designed payment scheme incentivizes agents to report their preferences truthfully.
However, there are settings where monetary transfers are not feasible, either because of ethical or legal issues \eat{or simply
  inappropriateness} \cite{AGT}, or because of practical issues with enforcing and collecting payments \cite{PT_nomoney}. This observation has led to a growing literature on designing mechanisms that incentivize agents to report their true preferences {\em without} using payments: for convenience, we refer to such mechanisms as {\em strategyproof}. 

In this paper, we focus on the design of strategyproof mechanisms for {\em assignment} problems. An instance of an assignment problem consists of a bipartite graph with items, or jobs, on one side, and bins, or machines, on the other; associated with each bin is a capacity, and with each edge a value and a size. A feasible assignment is a (partial) mapping from items to bins, where no bin's capacity is exceeded by the sizes of the items assigned to it. The goal is to compute a feasible assignment that maximizes \emph{welfare}: the sum of the values of the jobs for machines they are assigned to.  

The most general version of this problem, where both the size and value of a job can differ for different machines, is referred to as the generalized assignment problem (GAP). A number of well-known algorithmic assignment problems are special cases of the GAP. For instance, the problem with just one bin is the knapsack problem, and the problem with unit-sized items and bins is the maximum weight bipartite matching problem. Assignment problems are ubiquitous, and have been extensively studied due to their vast applicability, both from an algorithmic and mechanism design perspective. However, studying these problems in a setting {\em without money}, as we do, adds additional difficulties, since properties like monotonicity and weak monotonicity (see \cite{Myerson81,AT01,LMN03,AGT}) no longer suffice for truthfulness. Moreover, the popular VCG mechanism, as well as the maximal in range paradigm, require payments for truthfulness.

There are a number of settings where part of the input to an assignment problem is held by selfish agents, and the problem must be solved without the use of money. Unfortunately, as we discuss in \S \,\ref{sec:limits}, not much can be done if jobs hold their real-number values for the various machines private; this is consistent with many impossibility results from social choice theory. In this paper, we therefore focus on a restricted, yet natural, setting that admits interesting results--- for each pair $(i,j)$, it is public knowledge that the value of job $i$ for machine $j$ is either $v_{ij}$ or $0$, but job $i$ holds private which of those is the case for each $j$. This models situations where the private data encodes a {\em compatibility relation} between jobs and machines; the public $v_{ij}$ values arise in situations where the value derived from an assignment materializes over a public channel (for instance via a verifiable financial transaction). As a result, a job cannot hide its true value for any machine it anticipates being assigned to, although it can misreport a non-zero value as $0$, a lie that will not be discovered since the job will not be allocated to this machine. We will see that, if the mechanism is not chosen carefully, such strategic manipulations can be beneficial to a selfish agent even in very simple instances of the GAP. We also note that this is a \emph{multi-parameter} problem, where each job holds one bit of private information for each of the bins.

There are several natural settings that correspond to our model of private values. Suppose, for instance, a group of people are to split up a collection of tasks. Each task requires different skills, and how well each person can perform a task is public knowledge. Each task also has different time and location constraints, and whether or not the constraint is feasible for a person is only known privately to her. This problem is an instance of weighted bipartite matching; while finding an optimal solution is computationally easy, we are interested in algorithms that also ensure strategyproofness. In another example, consider a resource allocation problem, such as scheduling jobs on a collection of non-identical machines. The value of running a job on a machine, as well as the time it takes, is public knowledge.  However, each job requires specific  hardware and software that is available on only some of the  machines. Moreover, only the owner of the job knows which machines are compatible with the job. Here, the algorithm used to assign jobs to machines must ensure that the jobs do not have an incentive to lie about which machines are compatible.

As in \cite{PT_nomoney}, we are interested in strategyproof mechanisms that achieve a good {\em approximation} to the welfare of the optimal omniscient solution. The need for approximations arises for two reasons in our work: first, the GAP is NP-hard, as are several of its special cases; \ie, approximation is necessitated by computational intractability. The second reason is much more interesting--- unlike in settings with payments, solving the allocation problem optimally does not necessarily lead to a truthful mechanism, and we need to sacrifice approximation in order to obtain truthfulness.  We will see both factors playing a role as we seek strategyproof mechanisms that are good approximations to the various special cases of GAP. For example, we will see that no strategy proof mechanism for maximum-weight matching can be optimal, and we must sacrifice an approximation factor of $2$ for truthfulness. In contrast, a (non-polynomial time) algorithm that returns an optimal solution to the multiple knapsack problem while breaking ties consistently is strategyproof; however, since we are interested in polynomial-time strategyproof algorithms, we must resort to an approximately optimal mechanism that is both truthful and polynomial-time implementable.

Why would an agent benefit from lying in assignment problems when there are no payments? Consider, for instance, the weighted bipartite matching problem (\S \ref{sec:mwm}). Consider the following instance: job $a_1$ has edges with weights $1+\epsilon$ and $1$ to machines $b_1$ and $b_2$, and job $a_2$ has an edge to $b_1$ with weight $1$. An algorithm that simply chooses the maximum weight matching according to the reports incentivizes job $a_1$ to simply claim that the second edge does not exist: in the first case, the assignment chosen is $(a_1, b_2), (a_2, b_1)$, whereas in the second case, the assignment chosen is $(a_1, b_1)$, which suits $a_1$ better, with value $1+\epsilon$. This example makes it clear that that the optimal (and obvious) algorithms are not necessarily truthful--- not surprisingly, a carefully designed algorithm is essential to ensure that no agent has an incentive to lie, exactly as in mechanism design with money.

\subsection{Our Results}

We study the design of approximation mechanisms that are truthful without money for several variants of the GAP.  We begin in \S \,\ref{sec:combinatorial} with {\em matching}, which can be solved optimally in polynomial time from a  purely computational perspective.  We show that for the maximum matching problem, where all edge values are equal, simply returning the optimal solution while breaking ties consistently leads to a strategyproof mechanism. However, when a job's value depends on the machine, as in weighted bipartite matching, no deterministic strategyproof mechanism can achieve an approximation better than $2$; we provide such a mechanism.

Next, we examine knapsack-like variants of the GAP. Instead of specially tailored combinatorial algorithms for each variant, we extend the techniques in \cite{laviswamy} to reduce designing a truthful mechanism without money to designing such a mechanism for the {\em fractional} version of the problem: if the strategyproof mechanism for the fractional version yields an $\alpha$ approximation to the optimal fractional solution, and the corresponding LP has integrality gap $\beta$,  we derive a strategyproof randomized mechanism for the original problem with approximation ratio $\alpha\cdot \beta$. This technique applies to a large class of packing problems, and may of independent interest.

The GAP has integrality gap $2$, so a fractional strategyproof mechanism with approximation ratio $\alpha$ yields a $2\alpha$ strategyproof (in expectation)  mechanism for each of the knapsack-like variants of the GAP that we study. In \S \,\ref{sec:mkp}, we show, using network flows, that solving the fractional version of the multiple knapsack problem (MKP) optimally, while breaking ties consistently independent of the reported edges, gives an optimal strategyproof (fractional) mechanism. For size-invariant GAP (SIGAP), there is no optimal truthful (fractional) mechanism without money--- in \S \ref{sec:sigap}, we design a strategyproof, $2$-approximate greedy algorithm for fractional SIGAP. Using our extension of \cite{laviswamy} gives, respectively, $2$ and $4$-approximate strategyproof mechanisms for MKP and SIGAP. In \S \ref{sec:vigap_gap}, we sketch the construction of a $4$-approximate strategyproof mechanism for value-invariant GAP (VIGAP), as well as a $O(\log n)$-approximate strategyproof mechanism for the GAP. 

We point out that without the polynomial time restriction, there exist optimal strategyproof mechanisms for all variants of GAP where a node has the same value for each of its neighbors. That is, for maximum matching, MKP and VIGAP, simply solving the problem optimally, while breaking ties consistently independent of the private values (edges), leads to a truthful-without-money mechanism. For these problems, it is only computational intractability which causes us to lose an approximation factor. This is in contrast to the variants where a node has different values for different edges such as maximum weight matching and its generalizations. There, as we show in Theorem \,\ref{thm:mwm_lb}, strategyproofness and optimality cannot be achieved simultaneously.

\subsection{Related Work}Assignment problems have been studied extensively in the algorithms literature. Shmoys and Tardos \cite{ST} presented a 2-approximation for a minimization version of the \gap, and Chekuri and Khanna \cite{CK} observed that a 2-approximation to the maximization version -- the version considered in this paper -- is implicit in \cite{ST}. Moreover, it was shown in \cite{CK} that the multiple knapsack problem -- a special case of the \gap\ -- admits a PTAS\footnote{However, we note that this PTAS is not applicable to the generalization of \mkp\ that we consider, where assignments are constrained by a bipartite graph over items and bins. It follows from \cite[Theorem 3.2]{CK} that this is APX-hard.}, yet most generalizations of \mkp\ -- including the \gap -- are APX hard. Fleischer et al \cite{fleischer} obtained a $\frac{e}{e-1}$ approximation for the \gap, and showed that this is optimal for a slight generalization of the \gap. However, Feige and Vondrak \cite{FV06} then showed that the \gap\  admits a constant approximation slightly better than $\frac{e}{e-1}$, and this is the best currently known.

A number of results for the mechanism design version of assignment problems are known, although these are all in settings with money. In all of these results, the items hold their values private, and the rest of the instance is public. A $2$-approximate truthful-in-expectation mechanism follows immediately from the framework of Lavi and Swamy \cite{laviswamy}. Moreover, Briest et al \cite{BKV} devised a truthful FPTAS for the knapsack problem, as well as a truthful PTAS for \vigap\  when the number of bins is fixed. Recently, Azar and Gamzu \cite{AG} obtained a truthful 11-approximate mechanism for a variant of MKP, and Chekuri and Gamzu  obtained a $2+\epsilon$ approximation for a variant of VIGAP. We note that all the above mechanisms use money, and moreover the mechanisms in \cite{BKV,AG,CG} consider an incomparable setting to ours: our model is multi-parameter whereas theirs is single-parameter, but we consider a binary private value for each item and bin as opposed to an arbitrary real number.   

Mechanisms without money have a rich history in the social choice literature; for a survey, see \cite{AGT}. Interest in approximate mechanisms without money has been sparked by the recent work of Procaccia and Tennenholtz \cite{PT_nomoney}, which introduces the idea of using approximation to enable truthfulness in settings where solving optimally does not admit truthfulness without money. Approximate mechanisms without money have been developed for facility location \cite{PT_nomoney,AFPT_facility}, and selecting influential nodes in a social graph \cite{sumofus}.  In very recent work, Ashlagi et al.~\cite{Itai} study strategyproof mechanisms for matching motivated by kidney exchange. While we also study (bipartite) matching as a special case of the GAP, their model is very different from ours: each agent owns a set of vertices in a graph, and reports the existence of vertices to maximize the number of her vertices matched by the mechanism, plus the number that she can match amongst her hidden vertices and the vertices unmatched by the mechanism. Also related is the work of \cite{budish}, which studies a very general combinatorial assignment problem without money, and designs a mechanism which sacrifices efficiency, as well as weakens the notion of incentive compatibility, to achieve fairness.

\section{Model}
We describe the optimization version of the Generalized Assignment Problem (\gap), as well as its various special cases that we consider, in \S \ref{sec:gap}. We then review truthfulness in \S \ref{sec:truthfulness}, and discuss the limitations of truthfulness without money for the \gap\ in \S \ref{sec:limits}.  These limitations motivate our model of private valuations, which we then introduce in \S \ref{sec:privategraphmodel}.

\subsection{The Generalized Assignment Problem}\label{sec:gap}
In the \gap, there are $n$ jobs and $m$ machines. We  denote the set of jobs by $[n]=\set{1,\ldots,n}$, and the set of machines by $[m]=\set{1, \ldots, m}$. Machine $j$ has capacity $c_j \in \RR^+$. For each job $i$ and machine $j$, we associate a value $v_{ij} \in \RR^+$ and a size $s_{ij} \in \RR^+$.  An \emph{assignment} is a function $x : [n] \to [m] \union \set{*}$ partially mapping jobs to machines, where $*$ indicates that a job is left unassigned. We use $x(i)$ to denote the machine (or $*$) that job $i$ is assigned to. Moreover, the binary variable $x_{ij}$ indicates whether $x(i)=j$. A feasible assignment may allocate to machine $j$ a set of jobs of total size at most $c_j$.  The \gap\ can be written as an integer program with decision variables $\set{x_{ij}}_{ij}$; the LP relaxation obtained by relaxing the constraint $x_{ij} \in \{0,1\}$ is given below.

\begin{lp}
\tag{\gap\ LP}
\maxi{\sum_{i,j} v_{ij} x_{ij}}
\st \qcon{\sum_{j=1}^m x_{ij} \leq 1}{i=1,\ldots,n}
\qcon{\sum_{i=1}^n s_{ij} x_{ij} \leq c_j}{j=1,\ldots,m}
\con{0 \leq x_{ij} \leq 1}
\end{lp}

The above LP is known to have an integrality gap of $2$, and the rounding can be done in polynomial time \cite{ST,CK}. 

As detailed in \S \ref{sec:privategraphmodel}, we consider a setting where the private data is a bipartite graph specifying job-machine compatibility. That is, the private data are not the values $v_{ij}$ or the sizes $s_{ij}$, but rather the {\em existence} of the edge $(i,j)$. Note that this does not change the \gap\ from an algorithmic point of view, since one can encode the compatibility relation in the values $\set{v_{ij}}_{ij}$.
We define \gap[E] for a bipartite graph $E \sse [n] \cross [m]$
as the problem of computing the welfare-maximizing assignment using
only edges in $E$. The LP relaxation of \gap[E] is given below.

\begin{lp}
\tag{\gap[E] LP}
\maxi{\sum_{i,j} v_{ij} x_{ij}}
\st
\qcon{\sum_{j=1}^m x_{ij} \leq 1}{i=1,\ldots,n}
\qcon{\sum_{i=1}^n s_{ij} x_{ij} \leq c_j}{j=1,\ldots,m}
\con{0 \leq x_{ij} \leq 1}
\qcon{x_{ij} = 0}{(i,j) \notin E}
\end{lp}

We distinguish several variants of the \gap\ and their respective bipartite-graph versions. The \emph{Size-Invariant Generalized Assignment Problem} (henceforth \sigap) is the problem where the size of a job $i$ does not depend on the machine -- we denote this size by $s_i$. Similarly, the \emph{Value-Invariant Generalized Assignment Problem} (henceforth \vigap) is the problem where the value of a job $i$ does not depend on the machine -- we denote the value by $v_i$. The \emph{Multiple Knapsack Problem} (henceforth \mkp) is the problem where neither size nor value depend on the machine.  The \emph{knapsack problem} (henceforth \kp) is MKP with $m=1$.

In addition to knapsack-type problems like those dicussed above, the \gap\ also generalizes bipartite matching problems. The  \emph{maximum weight bipartite matching problem} (henceforth \mwbm) is the problem where all capacities and sizes are $1$. The \emph{maximum bipartite matching problem} (henceforth \mbm) is the special case of \mwbm\  where all values are $1$. The latter is only interesting when constrained by a graph; that is, when we consider \mbm[E] for some $E \sse [n] \cross [m]$.

When discussing a special case of \gap, say \mkp, we refer to the bipartite-graph constrained version as \mkp[E]. Moreover, we refer to (\gap\ LP) and (\gap[E] LP) as (\mkp\ LP) and (\mkp[E] LP), respectively. We also use similar notation for other special cases of the \gap.

\begin{figure}

\centering
\figeps[1]{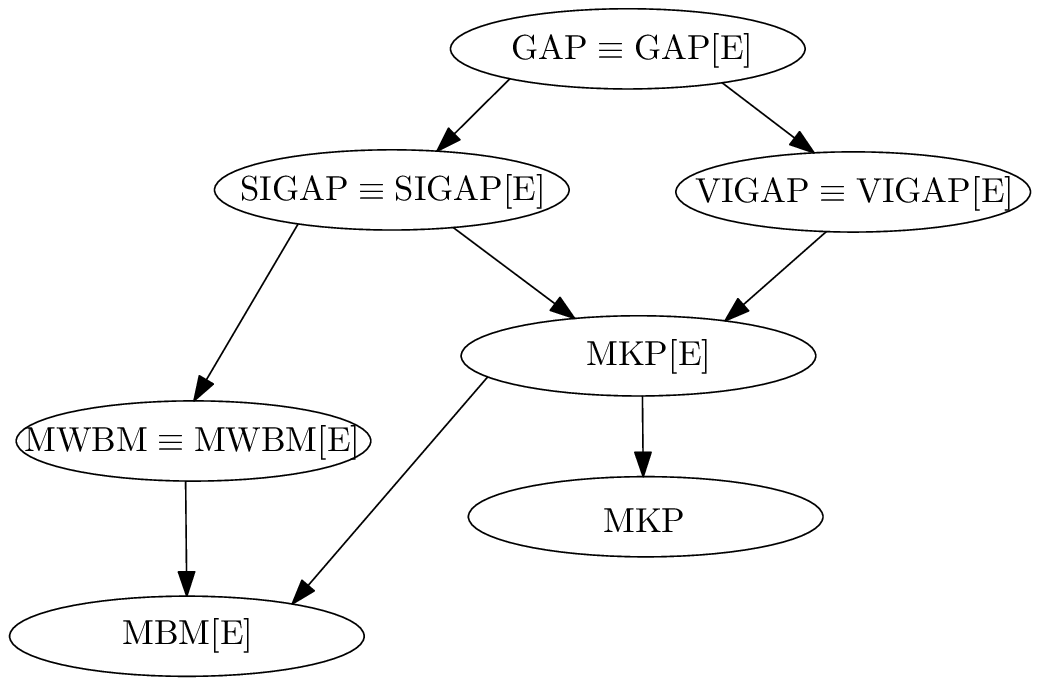}

\caption{Relationships between Assignment Problems}
Arrow indicates the second problem is a special
  case of the first. \\The $\equiv$ symbol indicates computationally
  equivalent problems.

\label{fig:prob_relationships}
\end{figure}

Figure \ref{fig:prob_relationships} illustrates the ordering of the various assignment problems by generality.
 We summarize the known algorithmic results, both upper and lower bounds, in Table \ref{table:computational}.

\begin{table}

\centering

{

\begin{tabular}{|c|c|c|c|}
\hline
{\bf Problem} & {\bf Upper } & {\bf Lower} & {\bf  Integrality} \\
{\bf } & {\bf Bound } & {\bf Bound} & {\bf  Gap} \\
\hline

\gap / \gap[E] & $\frac{e}{e-1} - \delta $ & APX-hard & $2$ \\

\hline

\sigap / \sigap[E] & $\frac{e}{e-1} - \delta $ & APX-hard & $2$ \\

\hline

\vigap / \vigap[E] & $\frac{e}{e-1} - \delta $ & APX-hard & $2$ \\

\hline

\mkp[E] & $\frac{e}{e-1} - \delta $ & APX-hard & $2$ \\

\hline

\mkp & PTAS & Strongly NP-hard & $2$ \\

\hline

\mwbm / \mwbm[E]   & $1 $ & $1$ & $1$ \\

\hline

\mbm[E]  & $1 $ & $1$ & $1$ \\

\hline

\end{tabular}
}

\caption{Computational Upper and Lower Bounds}
 $\delta$ is a small, positive, fixed real number.\\ 
   The integrality gap is given for the standard LP relaxation.
\label{table:computational}
\end{table}

\subsection{Truthfulness}\label{sec:truthfulness}
We consider the setting where jobs are selfish agents, and their private types encode information about their value for being assigned on different machines. Other information, such as $n$,$m$, $\set{s_{ij}}_{ij}$ and $\set{c_j}_{j=1}^m$ is considered public\footnote{It is conceivable that the job sizes $s_{ij}$ are private information as well. However, if jobs can lie about their sizes, we need to define the utility to a job when it is fractionally assigned--- a job can either derive fractional utility from a fractional assignment, or zero utility if it is not fully assigned to a machine. Both models are reasonable; for the first, a nontrivial proof shows that no reasonable approximation can be obtained by any randomized strategyproof mechanism. For the second model , where jobs derive no utility from partial assignments, it turns out that jobs have no incentive to lie about their sizes in any of the algorithms we design.}.

We assume that player $i$ has a type $v_i =\set{v_{ij}}_{j=1}^m$,
specifying his value for the different machines. We assume the
possible types of player $i$ are restricted to some public set $\V_i
\sse \RR^{m} $, and use $\V$ to denote $\V_1 \cross \ldots \V_n$. For
example, when working in the multiple knapsack problem, we require
that for each job $i$ there is a real number $v_i$ such that $v_{ij}= v_i$ for all
$j$. Moreover, as we will see in \S \ref{sec:limits}, no
interesting results without money are possible if possible values are
unbounded. Therefore, our positive results will assume a restricted,
discrete set of possible types described in \S \ref{sec:privategraphmodel}.

A \emph{mechanism without money} for the \gap\ is 
simply an algorithm that takes in all the problem data, public and private, and outputs
an assignment of the jobs to the machines. We allow our mechanisms to
be randomized. We use $\A(I,v)$ to denote the output of mechanism $\A$
on public data $I$ and private data $v \in \V$. Each mechanism $\A$
and instance of the public data $I$ induces a \emph{social choice
  rule} : a function $\A(I,*)$ mapping private valuations to
assignments.

We now state truthfulness without money generally.  

\begin{definition}
  Fix a mechanism without money $\A$, let $x$ denote the assignment $\A(I,v)$, and
  let $x'$ denote the assignment $\A(I,v_{-i} \union v'_i)$ when
  player $i$ changes his report to $v'_i$. Mechanism $\A$ is truthful
  if and only if the following always holds for for each $I$, $v \in \V$,
  $i$, and $v'_i \in \V_i$.
  \begin{equation}
    \label{eq:truthdet0}
    \sum_{j=1}^m v_{ij}
    x_{ij} \geq \sum_{j=1}^m v_{ij} x'_{ij}.
  \end{equation}
  Similarly, a randomized mechanism $\A$ is truthful in expectation without
  money if and only if
\begin{equation}
    \label{eq:truthexp0}
    \Ex[\sum_{j=1}^m v_{ij}
    x_{ij}] \geq \Ex[\sum_{j=1}^m v_{ij} x'_{ij}].
  \end{equation}
\end{definition}
In other words, a mechanism $\A$ is truthful if a job never benefits
from misreporting its private data. A randomized mechanism $\A$ is
truthful in expectation if no (risk-neutral) job has an incentive to
misreport its private data.

\subsection{Limits of Truthfulness without Money}\label{sec:limits}
Here, we will justify considering a discrete valuation model by
observing that, assuming general valuations, no interesting results
are possible. Indeed, we assume a restricted setting: the knapsack
problem, with a knapsack of capacity $1$, and jobs of size $1$. If
$\V_i = \RR^+$ for each $i$, then it is easy to see that this is
equivalent to the classical problem of a single item auction
\cite{Vickrey61}. It is well known that no non-trivial guarantees are
possible for a single-item auction if the mechanism is required to be
truthful-in-expectation without using money. In fact, it is easy to
see that no mechanism can outperform the trivial one which allocates
the item (in our case, the entire capacity of the knapsack) uniformly
at random, achieving an approximation ratio of $n$.

\subsection{The Private Graph Valuation Model}\label{sec:privategraphmodel}
Given that no nontrivial upperbounds
are possible when players can arbitrarily misrepresent their
values, we consider a restricted model of the
valuations: one of a discrete nature as is characteristic of many
problems for which truthfulness without money is possible. We assume
job $i$ has a value of $ \delta_{ij} v_{ij}$ for being assigned to
machine $j$, where $v_{ij}$ is public and $\delta_{ij} \in \set{0,1}$
is private. In other words, jobs may not lie about their potential value
$v_{ij}$ (that is, $v_{ij}$ are publicly known or {\em verifiable}), yet they may lie about which machines they are
\emph{compatible} with. This compatibility relation is encoded via a
bipartite graph on the jobs and machines. Each job's private data is the set of
its outgoing edges,  i.e. the machines with which it is
compatible. As we discuss in \S \ref{sec:intro}, this situation
arises in many natural settings. 

An \emph{instance} of the \emph{\gap\ on a private bipartite graph} is
a pair $(I,E)$, where $I$ is a tuple  $(\{v_{ij}\}_{ij}, \{s_{ij}\}_{ij},
\{c_j\}_j)$ of public information, and $E \sse [n] \times [m]$ is
private information solicited from the jobs. $E$ is a set of
\emph{edges} summarizing the compatibility of jobs and machines, where
job $i$'s private data is the set $E_i \sse E$ of edges in $E$
incident on $i$. A job $i$ receives value $v_{ij}$ from being assigned
to machine $j$ only if $(i,j) \in E$, else it receives value $0$.  Our
goal is to maximize welfare via a mechanism that, without using money,
incentivizes $i$ to report her set $E_i$ of edges truthfully.

We can now restate truthfulness as it applies to our model. We use
$\A(I,E)$ to denote the assignment computed by mechanism $\A$ on
instance $(I,E)$. Moreover, for an edge set $E$ we use $E_i \sse E$ to
denote the set of edges with one endpoint at job $i$, and use $E_{-i}$
to denote $E \sm E_i$.
\begin{definition}
  For a mechanism without money $\A$, let $x$ denote allocation
  $\A(I,E)$, and let $x'$ denote allocation $\A(I,E_{-i} \union
  E'_i)$.  $\A$ is truthful if and only if the following holds for
  each $I$, $E$, $i$, and $E'_i$.
  \begin{equation}
    \label{eq:truthdet1}
    \sum_{j: (i,j) \in E_i} v_{ij}
    x_{ij} \geq \sum_{j: (i,j) \in E_i} v_{ij} x'_{ij} 
  \end{equation}
  Similarly, a randomized mechanism $\A$ is truthful in expectation if and only
  if
\begin{equation}
    \label{eq:truthexp1}
    \Ex[\sum_{j: (i,j) \in E_i} v_{ij}
    x_{ij}] \geq \Ex[\sum_{j: (i,j) \in E_i} v_{ij} x'_{ij}]
  \end{equation}
\end{definition}
Note that the summation on both sides of the inequality is over the set of {\em true} edges $E_i$: that is, $\A$ is truthful [in expectation] if when a job misreports its
 incident edges, its [expected] utility from the new assignment $x'$ does not increase on its true edges $E_i$. 

\section{Combinatorial Mechanisms for Matching}

\label{sec:combinatorial}
We will show that optimally solving the maximum matching problem, with some careful tiebreaking, yields a strategyproof polynomial-time mechanism. For maximum weight matching, we show matching upper and lower bounds of 2 for truthful mechanisms without money, and a constant lowerbound for truthful-in-expectation mechanisms.

\subsection{Warmup: Maximum Bipartite Matching}
\label{sec:mm}
We consider the Maximum Bipartite Matching problem, constrained by a
private bipartite graph $E$. 
We observe that simply finding the maximum
matching, using consistent tiebreaking,  immediately gives a strategyproof mechanism. 

\begin{proposition}\label{prop:matching_tiebreaking_is_truthful}
  Fix a total order $\prec$ on matchings in the complete bipartite graph. For a set of edges $E$, let $M(E)$ denote the set of matchings on edge set $E$. Let $\A$ be the mechanism that, on input $(I,E)$, finds the $\prec$-minimal matching in the set $\argmax_{x \in M(E)}{\sum_{ij} x_{ij}}$.  Then $\A$ is truthful.
\end{proposition}
\begin{proof} 
 Assume for a contradiction that $\A$ is not truthful. Then, there exists $I$ and $E=(E_{-i}, E_i)$ and $E'_i$ violating \eqref{eq:truthdet1}. Let $x=\A(I,E)$ and $x'=\A(I,E')$, where $E'=E_{-i} \union E'_i$. Job $i$ is not matched by an edge in $E_i$ in $x$, yet is matched by an edge in $E_i$ in $x'$. Since $\A$ only uses reported edges, $i$ is not matched at all in $x$, yet is matched by an edge $e \in E_i \intersect E'_i$ in $x'$. This implies that both $x$ and $x'$ are in $M(E) \intersect M(E')$. Observe that $\sum_{ij} x_{ij} = \sum_{ij} x'_{ij}$, otherwise either $x$ is not optimal in $M(E)$, or $x'$ is not optimal in $M(E')$, contradicting the definition of $\A$. Therefore, both $x$ and $x'$ are optimal in both $M(E)$ and $M(E')$. Recalling that algorithm $\A$ breaks ties consistently, this yields a contradiction, as needed.
\end{proof}

Therefore, it suffices to define  $\prec$  so that the $\prec$-minimal maximum-matching on $E$ can be computed in polynomial time. This gives the following proposition.

\begin{proposition}\label{prop:mm}
  There is a polynomial-time, without-money mechanism for maximum bipartite matching that is optimal, and truthful in the private graph model.
\end{proposition}
\begin{proof}
  We represent each matching as a binary vector $(x_{11},x_{12},\ldots,x_{21},x_{22},\ldots,x_{nm})$ in $\set{0,1}^{nm}$, and let $\prec$ be the lexicographic order on these vectors. Then we proceed to find the $\prec$-minimal maximum matching as follows. We compute the size $OPT$ of the maximum matching (this can be done in polynomial time). We then process edges in the order they appear in the vector representation, while maintaining a working set of edges $X$ initialized to $E$. When processing an edge $e$, we check if removing $e$ decreases the size of the maximum matching in $X$ by solving the problem on edges $X \sm e$. If so we keep $e$ in $X$, else we discard $e$ by setting $X=X \sm e$. When finished, $X$ is a maximum matching; it is easy to see that $X$ is $\prec$-minimal among all maximum matchings.
\end{proof}

\subsection{Maximum Weight Bipartite Matching}
\label{sec:mwm}
Next, we consider the \mwbm\ problem, constrained by a private bipartite graph $E$.
Unlike MBM, we show constant lower bounds on the approximation ratio of truthful and truthful in expectation mechanisms.

\begin{figure}
\centering
\figeps[1]{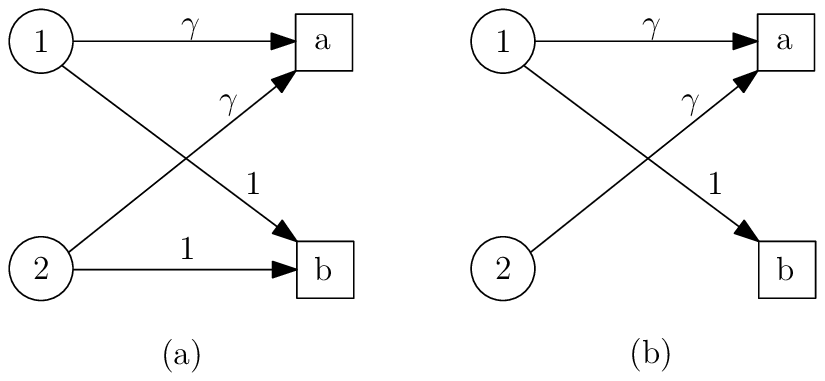}
\caption{Maximum Weight Matching Lowerbound}
Circles represent jobs, and squares represent machines.
\label{fig:mwm_lb}
\end{figure}

\begin{theorem}\label{thm:mwm_lb}
  No deterministic truthful mechanism without money for MWBM has approximation ratio better than $2$. Moreover, no truthful-in-expectation randomized mechanism gives better than a $\frac{2}{2 \sqrt{2} -1} \approx 1.0938$ approximation.
\end{theorem}
\begin{proof}
  First, consider a deterministic mechanism $\A$.  Assume for a contradiction that $\A$ attains an approximation $\alpha <2$. Consider Figure \ref{fig:mwm_lb}(a), where $\gamma >1$. Both jobs $1$ and $2$ prefer machine $a$ to machine $b$. If we let $\gamma$ be sufficiently close to $1$, $\A$ cannot leave either job unassigned. Without loss of generality, we assume job $1$ is assigned to $a$ and job $2$ is assigned to $b$. Now consider job $2$ changing his bid as in Figure \ref{fig:mwm_lb}(b). $\A$ cannot assign job $2$ at all under these new bids, as that would violate truthfulness. Therefore, when given the bids in Figure \ref{fig:mwm_lb}(b), the welfare of the solution is at most $\gamma$, whereas the optimal is $\gamma+1$. Letting $\gamma$ be sufficiently close to $1$ gives the contradiction.

  Now, we consider an arbitrary truthful-in-expectation mechanism $\A$ on Figure \ref{fig:mwm_lb}(a). At least one of the jobs must be assigned to the preferred machine $a$ with probability no more than $1/2$;  without loss of generality this is job $2$. Job $2$ derives value at most $(\gamma + 1) / 2$. Now, consider job $2$ changing his bid as in Figure \ref{fig:mwm_lb}(b).  By truthfulness, now $\A$ can assign job $2$ with probability at most $p=(\gamma + 1)/2\gamma$. Therefore, the welfare of the assignment returned  on the bids of Figure \ref{fig:mwm_lb}(b) is at most $\gamma + 1 - (1-p)\cdot 1 = \gamma + p$. However, the optimum is still $\gamma +1$, so the approximation ratio is at least \[ \frac{\gamma +1}{\gamma + p}= \frac{\gamma + 1}{\gamma + (\gamma+1)/2\gamma}\] Using elementary calculus, we can choose $\gamma$ to maximize this expression and complete the proof.
\end{proof}

Therefore, we cannot hope for better than a constant factor
approximation (specifically a PTAS, randomized or deterministic, is not
possible). We will show a factor $2$ deterministic truthful mechanism,
matching  Theorem \ref{thm:mwm_lb}.

\begin{algorithm}
\caption{Mechanism for \mwbm\ on Private Bipartite Graph}
\label{alg:mwbm}
\begin{algorithmic}[1]
\STATE Order pairs $(i,j) \in [n] \times [m]$ in decreasing order of
  $v_{ij}$, breaking ties arbitrarily.
\STATE Let $X= \emptyset$. 
\FORALL{$e \in E$ in the order defined above}
    \IF{$X \union \set{e}$ is a matching} \STATE Let  $X= X \union \set{e}$ \ENDIF
\ENDFOR
\RETURN $X$
 \end{algorithmic}
\end{algorithm}

Consider the greedy Algorithm \ref{alg:mwbm}. Notice that step (1)
{\em does not depend} on the reported edges $E$. 
\begin{theorem}
  Algorithm \ref{alg:mwbm} is a polynomial-time, $2$ approximate, no-money
  truthful mechanism for maximum weight bipartite matching in the private
  graph model.
\end{theorem}
\begin{proof}
  The approximation ratio immediately follows from a standard charging argument against the optimal solution.

  For truthfulness, consider a job $i$ misrepresenting his true edges $E_i$ as $E'_i$. Let
  $E' = E_{-i} \union E'_i$. Let $X$ be the matching returned by the
  algorithm on reports $E$, and let $X'$ be the matching returned on
  reports $E'$. If $X=X'$, then $i$ does not improve his
  value. Assume $X \neq X'$, and let $e' \in E'$ be the
  first edge in $X' \sm X$ according to the order of step (1). Since
  the algorithm processes edges in the bid-independent order of step
  (1), it is easy to see that $e' \in E' \sm E = E'_i \sm E_i$. Thus,
  $i$ is matched to $e' \neq E_i$, an edge from which he derives no value, when he reports $E'_i$. This completes the proof.
\end{proof}

\section{LP-Based Mechanisms for  Knapsack type Problems}\label{sec:lpbased}
The bipartite matching problems studied in the previous section can be solved in polynomial time; there, the need for approximation is a result purely of the requirement of strategyproofness. We now investigate knapsack-like variants of the \gap\ --- unlike matching, these problems are NP-hard (in fact, they are APX-hard). The mechanisms we design have the following common structure: first, we design a truthful mechanism without money for the {\em fractional} LP relaxation of the problem. Then, as in \cite{laviswamy}, we use a randomized procedure to obtain a feasible integral assignment from the fractional solution--- this composition leads to a truthful-in-expectation mechanism without money. We introduce this technique for designing truthful mechanisms without money in \S \ref{sec:fractional_reduction}, and then use it to design mechanisms without money for knapsack type assignment problems in \S \ref{sec:mkp}-\S \ref{sec:gap}. Some of our analyses will use notions from network flow theory, which we recap in Appendix \ref{app:flow}.

\subsection{A Reduction to Fractional Truthfulness}\label{sec:fractional_reduction}

The construction of Lavi and Swamy allows us to reduce constructing a truthful-in-expectation mechanism (with money) for the integral problem to constructing a truthful mechanism for a fractional version of the problem. We redevelop their construction, in a form convenient for our purposes, in Appendix \ref{app:laviswamy_recap}. While the truthful mechanism they construct for the  fractional \emph{welfare maximization problem} (defined in Appendix \ref{app:laviswamy_recap})  simply solves the problem optimally and uses VCG payments, we observe that this need not be the case. Indeed, this is crucial for our purposes; for some of the assignment problems we are interested in, the optimal algorithm for the fractional problem requires non-zero payments for truthfulness. Instead, we sacrifice optimality in the fractional solution to get a truthful fractional mechanism \emph{with zero payments}. Then we use the fractional mechanism to get a scaled down truthful-in-expectation mechanism -- also with zero payments -- for the combinatorial problem as in Theorem \ref{thm:laviswamy}.

By examining the proof of Theorem \ref{thm:laviswamy}, we notice that the assumption that the fractional mechanism solves the LP exactly is not used. In fact, an arbitrary truthful mechanism $M_{frac}$ for the fractional problem can be converted to a truthful-in-expectation mechanism $M_{exp}$ for the combinatorial problem. If $M_{frac}$ is a $\beta$-approximation algorithm for the fractional problem, then $M_{exp}$ is an $\alpha \cdot \beta$ approximation algorithm for the combinatorial problem, where $\alpha$ is the integrality gap of the LP relaxation. Moreover, by examing the proof of Theorem \ref{thm:laviswamy}, we notice that the payment scheme $p_{exp}$ of mechanism $M_{exp}$ is simply a scaled down copy of the payment scheme $p_{frac}$ of $M_{frac}$. In particular, if $M_{frac}$  is a truthful mechanism without money for the fractional problem, then $M_{exp}$ is a truthful-in-expectation mechanism without money for the combinatorial problem. We sum up these observations in the following Lemma.

\begin{lemma}\label{lem:fractional_suffices}
  Assume the fractional welfare maximization problem over polytope $P$ and valuation class $\V$ (as defined in Appendix \ref{app:laviswamy_recap}) admits a $\beta$-approximate mechanism that is truthful without money.  Moreover, assume $P$ satisfies the conditions of Lemma \ref{lem:laviswamy-containment} with integrality gap $\alpha$.  Then there exists an efficient truthful-in-expectation $\alpha \cdot \beta $-approximate mechanism for the welfare maximization problem over $P$ and $\V$ that does not use money.
\end{lemma}

In other words, we reduce the problem of designing a truthful-in-expectation mechanism without money to that of designing such a mechanism for its fractional relaxation. This will prove particularly useful,  since arguing about truthfulness of a continuous fractional assignment algorithm is more tractable than designing a combinatorial algorithm directly. Moreover, since the integrality gap of (\gap\ LP) is $2$,  an $\alpha$-approximate mechanism for a fractional assignment problem gives a $2 \alpha$-approximate mechanism for the integral problem. (Recall that the algorithm of \cite{ST,CK} shows an integrality gap of $2$ for \gap\  as needed for Lemma \ref{lem:laviswamy-containment}.)

\begin{corollary}\label{cor:gap_fractional_suffices}
  Consider any special case of the \gap. If the fractional version of the problem admits a $\beta$-approximate mechanism that is truthful without money, then there exists a $2 \beta$-approximate truthful-in-expectation mechanism for the combinatorial problem without money.
\end{corollary}

A note is in order on the LP relaxations of the \gap\ and various cases used in this reduction. Observe that the LP's for the variants of the \gap\ on a bipartite graph \emph{do not} fit the framework of \cite{laviswamy}. This is because the valuation -- i.e. the edges -- are encoded explicitly in the polytope and not in the objective. However, this is not a problem for us, since the equivalent \gap\ LP \emph{does} fit the framework. Therefore, after getting a fractional solution to, say, \mkp[E], we can simply re-interpret it as a fractional solution to (\gap\ LP) and perform the reduction of Lavi and Swamy.

\subsection{The Multiple Knapsack Problem}\label{sec:mkp}

We consider the multiple knapsack problem on a private bipartite graph $E$. 
First, we make the simple observation that, if we ignore computational constraints, there exists a truthful optimal mechanism for the multiple knapsack problem in the private graph model. As in maximum (unweighted) bipartite matching,  simply returning an optimal solution, breaking ties consistently, leads to a strategyproof mechanism. 

\begin{proposition}\label{prop:mkp_opt_is_truthful}
  Consider the without-money mechanism that, on reports $E \sse [n] \cross [m]$, finds the optimal integral solution to \mkp[E] LP, breaking ties consistently via an arbitrary total order $\preceq$ on the set of assignments $([m] \union \set{*})^{[n]}$. This mechanism is truthful in the private graph model.
\end{proposition}
\begin{proof}
  Fix a player $i$ with true edges $E_i$, and fix the reported edges $E_{-i}$ of the other players. Let $x$ be the assignment on reports $E=(E_i,E_{-i})$, and let $x'$ be the assignment on reports $E'=(E'_i,E_{-i})$. Assume for a contradiction that truthfulness is violated, in particular that $x(i)=*$, yet $x'(i)=j$ for some machine $j$ where $(i,j) \in E_i$.

  Recall that $x$ is the optimal feasible solution for $\mkp[E]$, and $x'$ is the optimal feasible solution for $\mkp[E']$. Since $i$ is unassigned in $x$, we know that $x$ is feasible for $\mkp[E']$ as well. Moreover, since $x'$ assigns $i$ using an edge in $E$, we know $x'$ is feasible for $\mkp[E]$. We conclude that each of $x$ and $x'$ is feasible and optimal for $\mkp[E]$ and $\mkp[E']$. This contradicts the consistent tie-breaking of the mechanism.
\end{proof}

The above implies that, unlike maximum weight matching and its various generalizations, \mkp\ is not fundamentally incompatible with truthfulness without money -- at least when ignoring computational constraints. Nevertheless, MKP on a bipartite graph is APX-hard. Therefore, we consider the problem of finding a truthful constant-factor approximation. Though a simple greedy algorithm gives a deterministic, truthful $2+\epsilon$ approximation, we will instead illustrate our techniques from Section \,\ref{sec:fractional_reduction} -- which will also come in handy for other generalizations of the \gap\ -- by designing a randomized, $2$-approximate, truthful-in-expectation mechanism for $\mkp$ in our model.

By the discussion in \S \ref{sec:fractional_reduction}, it suffices to devise a truthful fractional algorithm in the sense of Equation \eqref{eq:fractruth} of Appendix \,\ref{app:laviswamy_recap}. In this section, we show that solving (\mkp[E] LP) optimally, with careful tiebreaking, yields such an algorithm.  By Corollary \ref{cor:gap_fractional_suffices}, this yields a $2$-approximate truthful-in-expectation mechanism for MKP in the private graph model.

\begin{algorithm}
\caption{Fractional Mechanism for \mkp\ on Private Bipartite Graph}
\label{alg:mkpfrac}
\begin{algorithmic}[1]
  \REQUIRE Public instance of \mkp, and reported edges $E$.
  \ENSURE A feasible optimal solution $x$ for \mkp[E] LP
  \STATE Fix an arbitrary order on the edges of the complete bipartite
  graph $[n] \times [m]$. Let $\prec$ be the lexicographic order on
  $[\RR]^{[n] \times [m]}$ corresponding to the order on edges.
  \STATE Find the optimal solution $x$ to \mkp[E] LP, breaking ties
  according to $\prec$.
  \RETURN $x$
\end{algorithmic}
\end{algorithm}

Consider Algorithm \,\ref{alg:mkpfrac} for the fractional multiple knapsack problem on a bipartite graph. First of all, it is easy to see that Algorithm \,\ref{alg:mkpfrac} can be implemented in polynomial time by solving a sequence of linear programs in step (2).
In order to show truthfulness, we will use a flow-based interpretation of Algorithm \,\ref{alg:mkpfrac}. For reported edges $E$, we define graph $G[E]$, seen in Figure \ref{fig:mkp_flow}, as follows. There is a node for each job, and a node for each machine. We connect job $i$ to machine $j$ if $(i,j) \in E$, with weight $w_{(i,j)} = 0$ and capacity $c_{(i,j)} = \infty$. Next, we include a source node $s$, and create an edge $(s,i)$ for each job $i$ with weight $w_{(s,i)} = v_i / s_i$ and capacity $c_{(s,i)} = s_i$. We then create a sink $t$ and create an edge $(j,t)$ for each machine $j$, with weight $w_{(j,t)} = 0$ and capacity $c_{(j,t)} = c_j$. Finally, we connect the sink to the source via an edge $(t,s)$ with $w_{(t,s)} = 0$ and capacity $c_{(t,s)} = \infty$.

\begin{figure}
\centering
\figeps[1]{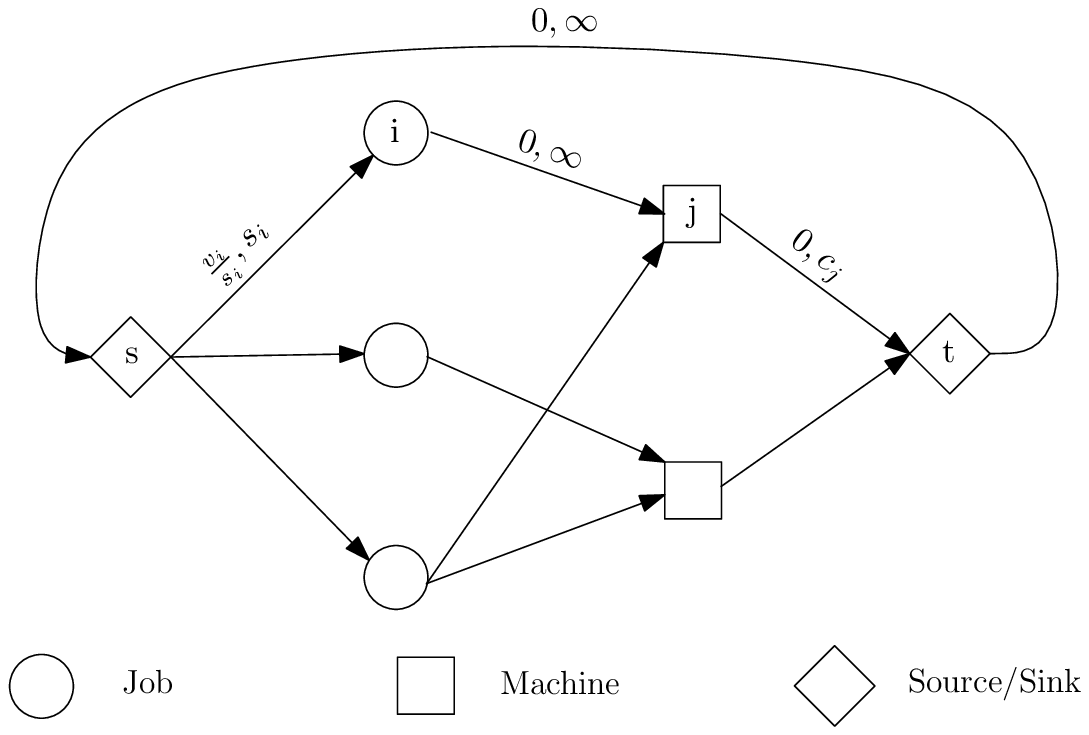}
\caption{Flow Interpretation of the Multiple Knapsack Problem}
Edges are labeled with (weight,capacity) pairs.
\label{fig:mkp_flow}
\end{figure}

Observe that fractional assignments  $[0,1]^{[n] \times [m]}$ are in one to one correspondence with feasible circulations in the complete bipartite graph $G[[n] \times [m]]$. Moreover, feasible solutions of \mkp[E] LP are in one-to-one correspondence with feasible circulations in $G[E]$. In particular, assignment $x$ feasible for \mkp[E] LP maps to the unique feasible circulation $f_x$ on $G[E]$ satisfying $f_x((i,j)) = x_{ij} s_i$ for each job $i$ and machine $j$. Notice, also, that the value (i.e. the welfare) of assignment $x$ is the same as the weight of the circulation $f_x$. Moreover, we define a total order on circulations in $G[[n] \times [m]]$ that corresponds to the lexicographic order $\prec$ defined on $[\RR]^{[n] \times [m]}$. We abuse notation and use $\prec$ to refer to both total orders, and let $f_x \prec f_y$ if and only if $x \prec y$. Notice that $\prec$ also orders feasible circulations lexicographically.  Therefore, we can interpret Algorithm \,\ref{alg:mkpfrac} as finding the $\prec$-minimal maximum-weight feasible circulation in $G[E]$, and then converting it to the corresponding assignment.

Next, we show that Algorithm \,\ref{alg:mkpfrac} is a truthful fractional mechanism, that is, a player $i$ cannot benefit by misrepresenting his edges $E_i$ as some $E'_i$. For this, if the algorithm returns assignment $x$ when reported edges are $E$, and assignment $x'$ when reported edges are $E'=E_{-i} \union E'_i$, then we must have $x(E_i) \geq x'(E_i)$ (note that this is because $v_{ij} = v_i$ in MKP). Equivalently, we need to show $\sum_{e \in E_i} f_x(e) \geq \sum_{e \in E_i} f_{x'} (e)$, where we use the convention $f(e) = 0$ when $f$ is a flow on $G[E]$ and $e \notin E$ (and the same for $G[E']$).

We will show that any increase in job $i$'s utility after lying implies that one of $f_x$ or $f_{x'}$ is suboptimal, yielding a contradiction. We will use notions from network flow, developed in Appendix \ref{app:flow}. We begin with the following lemma.

\begin{lemma}\label{decompose-difference}
  Let circulation $\Delta$ be the difference between $f_{x'}$ and $f_x$; i.e. $\Delta=f_{x'}-f_x$. If $\sum_{e \in E_i} f_{x'}(e) > \sum_{e \in E_i} f_x(e)$ then $\Delta$ can be conformally decomposed into $\set{C,\Delta - C}$ where: (1) $C$ is a flow cycle, and (2) $C$ sends positive flow on both $(s,i)$ and some $e \in E_i
    \intersect E'_i$. 
\end{lemma}
\begin{proof}
  Let $\set{C^1,\ldots,C^k}$ be the conformal decomposition of $\Delta$ into cycles as in theorem \ref{conformal-decomposition-cycles}.  It suffices to show that some $C^j$ satisfies conditions (1) and (2).

  Since $x'$ sends more flow on edges in $E_i$ than $x$, some $C^j$ enters $i$ through an edge $e' \notin E_i$, and exits through an edge $e \in E_i$. By conformality and the fact that $x'$ sends no flow on $E_i \sm E'_i$, we know that $e \in E_i \intersect E'_i$.  Moreover, by conformality and the fact that $x$ sends no flow on edges in $E'_i \sm E_i$, it is easy to see that $e' \notin E'_i \sm E$.  Therefore, the only remaining possibility is that $e' = (s,i)$.
  \end{proof}
  \noindent This yields truthfulness of the algorithm.
\begin{lemma}
  Algorithm \,\ref{alg:mkpfrac} is a truthful fractional mechanism
\end{lemma}
\begin{proof}
   Fix an instance $(I,E)$, and assume for a contradiction that a player $i$ with true edges $E_i$ benefits by reporting $E'_i$ instead. Let $x$ and $x'$ be the assignments computed by the algorithm on reports $E$ and $E' = E_{-i} \union E'_i$. By assumption, $\sum_{e \in E_i} f_{x'}(e) > \sum_{e \in E_i}f_x(e)$. Let $\Delta$ and $C$ be as in Lemma \ref{decompose-difference}. Observe that $C$ does not send flow on any edges in the symmetric difference of $E_i$ and $E'_i$. Therefore, by Lemma \ref{lem:feasibility_maintained} $f_x+C$ is a feasible circulation in $G[E]$, and moreoever $f_{x'}-C$ is a feasible circulation in $G[E']$.  If the weight $w(C)$ of circulation $C$ is non-zero, then one of $f_x$ or $f_{x'}$ is non-optimal. Therefore, $w(C)=0$. Now notice that, by definition of $\prec$, either $f_x+C \prec f_x$ or $f_{x'} - C \prec f_{x'}$. Therefore, one of $f_{x}$ or $f_{x'}$ is not a $\prec$-minimal optimal solution, yielding the contradiction.
\end{proof}

\noindent Combining with Corollary \,\ref{cor:gap_fractional_suffices}, we get the following theorem.
\begin{theorem}
  There is a polynomial-time, $2$-approximate, without-money mechanism for the multiple knapsack problem that is truthful-in-expectation in the private graph model.
\end{theorem}

\subsection{Size-Invariant \gap}
\label{sec:sigap}

We consider the size-invariant generalized assignment problem on a private bipartite graph $E$. 
Since \sigap\ generalizes \mwbm, by Theorem \,\ref{thm:mwm_lb} no deterministic truthful approximation can achieve better than a factor $2$ approximation, and moreover no truthful in expectation PTAS is possible. In this section, we devise a 4-approximate without-money mechanism for \sigap\ that is truthful-in-expectation in the private graph model. Even though solving \sigap[E] LP is not fractionally truthful (again, by Theorem \ref{thm:mwm_lb}), we show that a simple greedy algorithm is fractionally truthful and yields a 2-approximate solution the LP. Combining this with Corollary \,\ref{cor:gap_fractional_suffices}, we get a $4$-approximate, without-money mechanism for \sigap\ that is truthful in expectation in the private graph model.  Consider the following algorithm.

\begin{algorithm}
\caption{Fractional Mechanism for \sigap\ on Private Bipartite Graph}
\label{alg:sigapfrac}
\begin{algorithmic}[1]
  \REQUIRE Public instance of \sigap, and solicited private edges $E$.
   \ENSURE A feasible solution $x$ for \sigap[E] LP
   \STATE Order  $[n] \times [m]$ in decreasing order 
     of value density $d_e$, where $d_{(i,j)}
     =\frac{v_{ij}}{s_i}$, breaking ties arbitrarily.
  \FORALL{ $(i,j) \in E$, in the order defined above}
   \STATE Fractionally assign as much of job $i$ on machine $j$, until
      the job is exhausted or the machine is full.
    \ENDFOR
  \RETURN the resulting assignment $x$.

\end{algorithmic}
\end{algorithm}

First, we bound the approximation factor of Algorithm \,\ref{alg:sigapfrac}.
\begin{lemma}\label{lem:sigap_approx}
  Algorithm \,\ref{alg:sigapfrac} returns a $2$-approximate solution
  to \sigap[E] LP.
\end{lemma}
\begin{proof}
  This can be shown by a charging argument, best formalized by constructing a feasible solution to a dual of \sigap[E] LP of value at most twice the value attained by the algorithm. This dual is shown below, and has decision variables $u \in \RR^n$ and $z \in \RR^m$:

\begin{lp}
 \tag{\sigap[E] LPD}
 \mini{\sum_{i=1}^n u_i + \sum_{j=1}^m c_j z_j}
\st
\qcon{u_i + s_i z_j \geq v_{ij}}{(i,j) \in E}
\qcon{u_i \geq 0}{i=1,\ldots,n}
\qcon{z_j \geq 0}{j=1,\ldots,m}
\end{lp}
 First, we make a simple observation about the algorithm that will be useful in the proof.
  \begin{observation}\label{obs:edgeorder}
    For a job $i$, edges incident on $i$ are examined in decreasing order of $v_{ij}$ (since size $s_i$ is independent of the machine). For a machine $j$, edges incident on $j$ are examined in decreasing order of $v_{ij} / s_i $.
  \end{observation}
We now construct the dual solution $u,z$ in parallel with the execution of the algorithm as follows. Begin with $u=0$ and $z=0$. Consider the iteration of Algorithm \,\ref{alg:sigapfrac} corresponding to edge $e=(i,j)$. If job $i$ is exhausted on this iteration, set $u_i = v_{ij}$. If the capacity on machine $j$ is exhausted on this iteration, set $z_j = v_{ij} / s_i$.  Notice that, in both cases, this satisfies the dual constraint corresponding to edge $(i,j)$. If no assignment is made on this iteration -- i.e. either $i$ or $j$ was exhausted in a previous iteration -- then we do not update the dual variables. Indeed, there is no need to do so for feasibility: By Observation \ref{obs:edgeorder}, if $i$ is already exhausted then already $u_i \geq v_{ij}$, and if $j$ is already exhausted then already $z_j \geq v_{ij} / s_i$, and either suffices to satisfy the dual constraint for edge $e$.

It remains to bound the value of the dual solution as compared to the primal solution. First, we write twice the value of the primal in a convenient form:
\[   2 \sum_{i,j} v_{ij} x_{ij} = \sum_i \sum_j v_{ij} x_{ij} + \sum_j
\sum_i \frac{v_{ij}}{s_i} (s_i x_{ij})\]

Observe that, by  Observation \ref{obs:edgeorder}, $u_i$ lower bounds the value of any edge on which any part of job $i$ is assigned, and $z_j$ lower-bounds the density of any job assigned to machine $j$. Moreover, $u_i$ is non-zero only if $i$ is fully assigned, and $z_j$ is non-zero only if $j$ is full. Therefore, we get
\begin{align*}
\sum_i \sum_j v_{ij} x_{ij} + \sum_j \sum_i \frac{v_{ij}}{s_i} (s_i
x_{ij}) &\geq \sum_i u_i \sum_j x_{ij} + \sum_j
z_j \sum_i (s_i x_{ij}) \\
&= \sum_i u_i + \sum_j z_j c_j
\end{align*}

The final term is precisely the value of the dual. Invoking weak LP duality completes the proof.
\end{proof}

It remains to show that Algorithm \,\ref{alg:sigapfrac} is fractionally truthful. We begin with an observation.

\begin{observation}\label{obs:lexicographically_maximal}
  Let $e_1,\ldots,e_{nm}$ be the ordering $[n] \times [m]$ in decreasing order of density $v_{ij} / s_i$.  Let $\prec$ denote the lexicographic ordering on $\RR^{[n] \times [m]}$.  Algorithm \,\ref{alg:sigapfrac} returns the $\prec$-maximal feasible solution of \sigap[E] LP.
\end{observation}

\begin{figure}
\centering
\figeps[1]{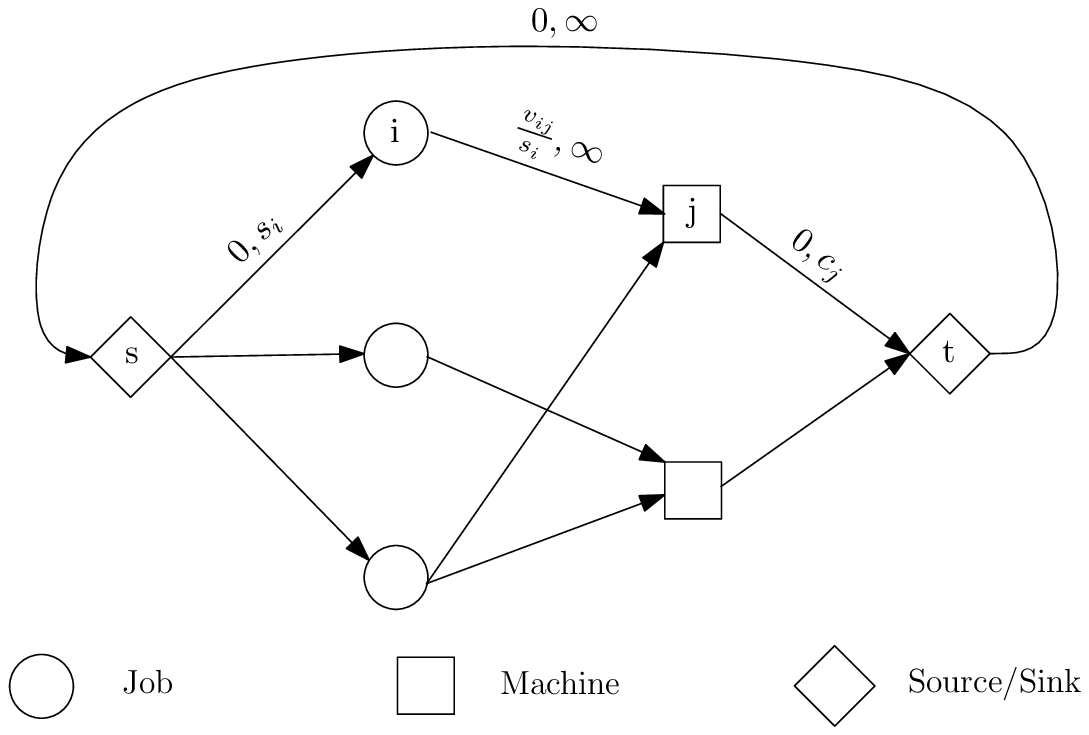}
\caption{Flow Interpretation of the Size-invariant Generalized
  Assignment Problem}
Edges are labeled with (weight,capacity) pairs.
\label{fig:sigap_flow}
\end{figure}

As in section \ref{sec:mkp}, we use a network flow interpretation. We define graph $G[E]$ for  $E \sse [n] \times [m]$. This construction is similar to that of \S \ref{sec:mkp}, and the graph is shown in Figure \ref{fig:sigap_flow}. 
As in \S \ref{sec:mkp}, feasible assignments for \sigap[E] LP are in one-to-one correspondence with feasible circulations in $G[E]$. We can similarly define an order $\prec$ on flows in $G[[n] \times [m]]$ by setting $f_x \prec f_{y}$ when $x \prec y$. Therefore, we can interpret Algorithm \,\ref{alg:sigapfrac} as finding the $\prec$-maximal feasible circulation in $G[E]$. Next, we establish a decomposition lemma similar to, yet more involved than, Lemma \ref{decompose-difference}.

\begin{lemma}\label{decompose-difference-sigap}
  Let $x$ be the output assignment of Algorithm \,\ref{alg:sigapfrac} with declarations $E$, and let $x'$ be the output assignment with declarations $E' = E_{-i} \union E'_i$. Let circulation $\Delta$ be the difference between $f_{x'}$ and $f_x$; i.e. $\Delta=f_{x'}-f_x$. If $\sum_{e \in E_i} w_e f_{x'}(e) > \sum_{e \in E_i} w_e f_x(e)$ then $\Delta$ can be conformally decomposed into $\set{C,\Delta - C}$ where:
  \begin{enumerate}
  \item $C$ is a flow cycle \label{cond:is_cycle}
  \item $C \prec 0$. In other words, adding $C$ to any other
    circulation worsens the lexicographic order.\label{cond:is_worsening}
  \item $C$ enters $i$ through some $e \in E_i \sm E'_i$, with
    $C(e) <0$. \label{cond:enters}
  \item $C$ exits $i$ through some $e' \in E_i \intersect E'_i$,
    with $C(e') > 0$. \label{cond:exits}
  \item $w_{e'} > w_e$. \label{cond:weights}
  \end{enumerate}
\end{lemma}
\begin{proof}
  Consider the conformal decomposition of $\Delta$ into cycles $C^1,\ldots,C^k$. It suffices to show that some $C^j$ satisfies the conditions above. By assumption, there is a cycle $C$ in the decomposition with $\sum_{e \in E_i} w_e C(e) > 0$. By conformality, cycle $C$ must exit $i$ through an edge $e' \in E_i \intersect E'_i$. Moreover $C$ cannot enter through an edge in both graphs $G[E]$ and $G[E']$: if it did, then both $f_{x} + C$ and $f_{x'} -C$ are feasible in $G[E]$ and $G[E']$ respectively, and thus one of them is not lexigographically maximal, contradicting observation \ref{obs:lexicographically_maximal}. Therefore, by conformality, $C$ must remove flow from an edge $e \in E_i \sm E'_i$, and add flow to $e' \in E_i \intersect E'_i$, with $w_{e'} > w_{e}$. Thus, $C$ satisfies conditions \ref{cond:is_cycle}, \ref{cond:enters}, \ref{cond:exits}, and \ref{cond:weights}.

  It remains to establish condition \ref{cond:is_worsening}. Observe that $f_x + C$ is a feasible circulation on $G[E]$. Since $f_x$ is a lexicographically maximal feasible circulation on $G[E]$, we deduce $C \prec 0$.
\end{proof}
We are now ready to show truthfulness.

\begin{lemma}
  Algorithm \,\ref{alg:sigapfrac} is a truthful fractional mechanism for \sigap.
\end{lemma}
\begin{proof}

  Assume, for a contradiction, that a player $i$ benefits by reporting $E'_i$ instead of his true edges $E_i$. Let $x$ and $x'$ be the assignments computed by the algorithm on reports $E$ and $E' = E_{-i} \union E'_i$. By assumption, $\sum_{e \in E_i} w_e f_{x'}(e) > \sum_{e \in E_i}w_e f_x(e)$. Let $\Delta$ and $C$ be as in Lemma \ref{decompose-difference-sigap}.

  Recall that $C \prec 0$. Thus, if $f_{x'} - C$ were feasible in $G[E']$ we would be done, as we would contradict Observation \ref{obs:lexicographically_maximal}. However, this is not the case, as $f_{x'} - C$ sends positive flow on edge $e \in E_i \sm E'_i$. We remedy this by simply zeroing out the flow on edge $e=(i,j)$, as follows: Let $D$ be the flow cycle through $s,i,j,t$ such that $f_{x'} -C -D$ sends no flow on $e$. It is clear that $f_{x'} - C - D$ is a feasible circulation on $G[E']$.  We claim that still $f_{x'} -C -D \succ f_{x'}$. To see this, notice that by condition \ref{cond:weights} of Lemma \ref{decompose-difference-sigap}, edge $e$ is not the greatest weight edge with non-zero flow in $-C$. Thus, since $-C \succ 0$, it is easy to see that also $-C - D \succ 0$. Therefore $f_{x'} -C - D \succ f_{x'}$, as needed.
\end{proof}
Combining with Corollary \,\ref{cor:gap_fractional_suffices}, we get
the theorem.

\begin{theorem}
  There is a polynomial-time, $4$ approximate, without-money mechanism for the size-invariant generalized assignment problem that is truthful-in-expectation in the private graph model.
\end{theorem}

\subsection{\vigap\ and \gap}
\label{sec:vigap_gap}
In this section, we overview the results that we obtain for \vigap\
and \gap, utilizing the techniques developed in \S
\ref{sec:fractional_reduction}-\ref{sec:sigap}. However, since these 
differ from our results so far only technically, we  defer details to the full version of the paper.

Consider the \vigap. We observe that Proposition \ref{prop:mkp_opt_is_truthful} holds essentially unchanged; that is, solving \vigap\ optimally gives a truthful mechanism in our model. Next, we observe that greedy Algorithm \,\ref{alg:sigapfrac}, when adapted to the fractional \vigap\ (namely, density $d_{(i,j)}$ is now defined as $v_{i}/s_{ij}$),  still provides a $2$-approximation by essentially the same analysis. A more involved inductive argument is needed to show that this fractional algorithm is truthful; We defer this technical, yet simple, proof to the full version of the paper. We get the following theorem.

\begin{theorem}\label{thm:vigap}
  There is a polynomial-time, $4$ approximate, without-money mechanism for the value-invariant generalized assignment problem that is truthful-in-expectation in the private graph model.
\end{theorem}

We now turn to the \gap. First, we make an assumption -- to be removed later -- that the maximum value $v_{max}$ of an edge in $E$ is publicly known up-front. Under this assumption, we reduce designing a truthful mechanism for \gap\ to the truthful mechanism for \vigap\, with a loss of $O(\log n)$ in the approximation ratio. In particular, we randomly pick $v \in \{v_{\max},\frac{v_{\max}}{2}, \frac{v_{\max}}{4}, \ldots, \frac{v_{\max}}{O(n^2))}\}$, and define a new value $\hat{v}_{ij}$ for each item $i$ and bin $j$ as follows: If $v_{ij} \geq v$ then $\hat{v}_{ij}=v$, else $\hat{v}_{ij}=0$ (equivalently, we discard edge $(i,j)$). This gives an instance of \vigap\ that we can solve using the mechanism of Theorem \ref{thm:vigap}. It is easy to verify that, in expectation, this reduction results in a loss of at most $O(\log n)$ in the approximation ratio. However, to ensure truthfulness we need to guarantee that edge $(i,j)$ actually results in value exactly $\hat{v}_{ij}$ if chosen; we do so by positing up-front that, whenever job $i$ is assigned to machine $j$ by the subroutine that solves \vigap, we cancel $i$'s assignment with probability $1-\hat{v}_{ij}/v_{ij}$. It is now easy to see that, under our original assumption that $v_{\max}$ is known up-front, this mechanism is truthful-in-expectation and has an approximation ratio of $O(\log n)$.

We now remove the assumption that $v_{max}$ is public knowledge by appropriately incentivizing the job with the maximum value edge. In particular, after receiving the reported edges $E$, we flip a fair coin. If the coin turns up heads, we assign the job with the maximum value edge on that edge (\ie, to his favorite machine, with value $v_{\max}$), and leave all other jobs unassigned. If the coin turns up tails, we discard the job with the maximum value edge, and proceed with the algorithm described above using this value of $v_{max}$. It is easy to see that this is still an $O(\log n)$ approximation algorithm. Moreover, the job with the maximum value edge can do no better than report his true edges, and no job has incentive to falsely claim a maximum value edge. This gives the following theorem.

\begin{theorem}\label{thm:gap}
  There is a polynomial-time, $O(\log n)$ approximate, without-money mechanism for the generalized assignment problem that is truthful-in-expectation in the private graph model.
\end{theorem}

We leave open the question of whether there exists a constant factor truthful [in expectation] mechanism without money for the \gap\ in our model.

\section{Acknowledgements}
We thank Preston McAfee, Serge Plotkin, Ariel Procaccia, Tim Roughgarden, Michael Schwarz, and Mukund Sundararajan for helpful discussions and comments.

{
\bibliography{nomoney}
\bibliographystyle{plain}
}
\newpage
\appendix
\section{ The Construction of Lavi and Swamy }\label{app:laviswamy_recap}
First, we recall some basic concepts from combinatorial optimization, and state the key lemma in the construction of Lavi and Swamy.

\begin{definition}
  Let $P \sse \RR^d$ be a polytope. We define the \emph{Integer hull} of $P$, which we denote by $I(P)$, as the convex hull of all integer points in $P$. Equivalently, \[I(P) = hull(P \intersect \ZZ^d)\]
\end{definition}

\begin{definition}
  Let $P$ be a polytope. The \emph{integrality gap} of $P$ is defined as:
\[ \max_{c \in \RR^d} \frac{\max\{c^Tx : x \in P\}}{\max\{c^Tx : x \in I(P)\}} \]
\end{definition}

Note that the integrality gap of a polytope does not depend on any particular set of objectives $c$. In fact, even in problems where the only objectives of interest are of a certain class -- say submodular valuations as used in combinatorial auctions, or objectives $c$ in the nonnegative orthant -- the construction of Lavi and Swamy can only perform as well as the integrality gap of the polytope as a whole, as defined above.

\begin{definition}
  We say an algorithm \emph{shows an integrality gap of $\alpha$} for a polytope $P \sse \RR^d$ if it takes as input an arbitrary vector $c \in \RR^d$, and outputs a point $z \in P \intersect \ZZ^d$ with the guarantee that:
\[c z \geq \frac{1}{\alpha} \max\{cx : x \in P \}\]
\end{definition}

The construction of Lavi and Swamy concerns \emph{packing polytopes}. A polytope $P \sse \RR^d$ is a packing polytope if it is contained in the positive orthant -- i.e. $P \sse \RR+^d$ --, and moreover it is \emph{downwards closed}: if $y \in P$ and $x \in \RR+^d$ is such that $x \preceq y$ (component-wise) then $x \in P$.

Now, we come to the key lemma. We consider a packing polytope $P$, and use $P/\alpha$ to denote the scaled down copy of $P$ -- i.e. $P/\alpha = \set{ y /\alpha: y \in P}$. Using an algorithm showing the integrality gap of $\alpha$ for $P$, we can explicitly construct for every $x \in P/\alpha$ a distribution over integer points of $P$ that evaluates to $x$ in expectation.

\begin{lemma}[\cite{laviswamy}] \label{lem:laviswamy-containment} Let $P$ be a packing polytope of integrality gap at most $\alpha$. Then, for every $x \in P/\alpha$ there exists a distribution $D_x$ over $P \intersect \ZZ^d$ such that $E_{y \sim D_x} y = x$. Moreover, if there exists a polynomial-time algorithm $B$ that shows an integrality gap of $\alpha$ for $P$, then, for any $x \in P/\alpha$ we can compute $D_x$ in polynomial time.
\end{lemma}

Armed with the above lemma, Lavi and Swamy \emph{reduce} designing a truthful-in-expectation mechanism (with money) to designing a \emph{truthful fractional mechanism} (also with money), to be defined later.

We illustrate the technique of Lavi and Swamy by considering welfare maximization problems in a very general form. Let $P \sse \RR+^d$ be a packing polytope representing the set of feasible solutions. Moreover, assume that there are $n$ players $[n]$, and the \emph{valuation function} $v_i: \RR^d \to \RR$ of player $i$ is required to lie in some valid set of linear valuations $\V_i \sse \RR^{\RR^d}$. We denote $\V=\V_1 \cross \V_2 \ldots \cross \V_n$. The \emph{combinatorial welfare maximization problem} (henceforth CWMP) over $P$ and $\V$ is the problem of finding an integer point in $P$ maximizing the sum of values of the players. This gives the following LP relaxation, which we refer to as the \emph{fractional welfare maximization problem} (henceforth FWMP):

\begin{lp*}
\maxi{ \sum_i v_i(x)}
\st
\con{x \in P}
\end{lp*}

We define a \emph{truthful-in-expectation mechanism} for the combinatorial welfare maximization problem as a pair $M=(f,p)$ where $f:\V \to P \intersect \ZZ^d$ is a \emph{randomized allocation rule}, and $p: \V \to \RR^n$ is a \emph{randomized payment scheme}. As usual, we say $M$ is \emph{truthful in expectation } if the following holds for each $\vec{v} \in \V$ and $v_i \in \V_i$

\begin{align}
\Ex[v_i(f(v)) - (p(v))_i] \geq  \Ex[v_i(f(v_{-i},v'_i)) - (p(v,v'_i))_i].
\end{align}

Moreover, we define a \emph{fractional mechanism} as a pair $M=(f,p)$ where $f: \V \to P$ is a \emph{fractional allocation rule}, and $p: \V \to \RR^n$ is a \emph{payment scheme}. We say $M$ is a \emph{truthful fractional mechanism} if the following holds for each $\vec{v} \in \V$ and $v_i \in \V_i$

\begin{align}\label{eq:fractruth}
v_i(f(v)) - (p(v))_i \geq  v_i(f(v_{-i},v'_i)) - (p(v,v'_i))_i. 
\end{align}

It is easy to see that we can use the VCG mechanism to obtain a truthful fractional mechanism $M_{frac}$ for the fractional welfare maximization problem. Moreover, when optimizing objectives in $\V$ over $P$ can be done efficiently, $M_{frac}$ is a polynomial time mechanism. Lavi and Swamy observed that, given an efficient approximation algorithm that shows the integrality gap of $P$ in the sense defined above, one can obtain a polynomial-time, truthful-in-expectation mechanism $M_{exp}$ that looks simply like a scaled down version of $M_{frac}$. This follows from lemma \ref{lem:laviswamy-containment}. We state the main theorem of Lavi and Swamy and sketch its proof below.

\begin{theorem}[\cite{laviswamy}]\label{thm:laviswamy}
  Assume the fractional welfare maximization problem can be solved efficiently for any valuations in $\V$. Moreover, assume $P$ satisfies the conditions of Lemma \ref{lem:laviswamy-containment} with integrality gap $\alpha$ and algorithm $B$ showing the integrality gap. Then there exists an efficient truthful-in-expectation $\alpha$-approximate mechanism for the welfare maximization problem over $P$ and $\V$.
\end{theorem}
\begin{proofsketch}
  Consider the efficient, truthful fractional mechanism $M_{frac}= (f_{frac},p_{frac})$ constructed using VCG. We will define $M_{exp} = (f_{exp},p_{exp})$ as follows. We let $p_{exp}(v) = \frac{p(v)}{\alpha}$, and let $f_{exp}(v)$ be drawn from the distribution $D_{f_{frac}(v)/\alpha}$, as defined in Lemma \ref{lem:laviswamy-containment}. The random function $f_{exp}$ can be evaluated efficiently using $B$, as in the lemma. Now, using linearity of expectations, the linearity of the valuation functions, and Lemma \ref{lem:laviswamy-containment}, it is easy to show truthfulness of $M_{exp}$ by showing that the expected valuation less the payment of a player matches that of $M_{frac}$ up to a scaling factor of $\alpha$.
\begin{align*}
  \Ex[v_i(f_{exp}(v')) - (p_{exp}(v'))_i] &= v_i(\Ex[f_{exp}(v')]) -
  (p_{exp}(v'))_i\\
  &= v_i(\frac{1}{\alpha} \cdot f_{frac}(v')) - \frac{1}{\alpha} \cdot
  (p_{frac}(v'))_i \\
  &= \frac{1}{\alpha} \cdot [v_i(f_{frac}(v')) - (p_{frac}(v'))_i]
\end{align*}
Applying this equivalence to both sides of inequality \eqref{eq:fractruth} yields truthfulness of $M_{exp}$.
\end{proofsketch}
It is worth noting that a simple modification to $p$ allows us to guarantee individual rationality, while preserving truthfulness-in-expectation. However, this will not be relevant for our purposes, as all our mechanisms will utilize no payments.

\section{Network Flow Preliminaries}
\label{app:flow}
Here we recall some notions from network flow theory, and define simple
concepts that we will need in analysis of our
algorithms for MKP and SIGAP. 

We recall that a \emph{weighted-capacitated directed
  graph} is a directed graph $G=(V,E)$, with a weight $w_e \in \RRp$
and capacity $c_e \in \RRp$ on an edge $e$. We define a flow simply as
follows.

\begin{definition}
  A flow is a vector $f \in \RR^{E(G)}$.
\end{definition}

We distinguish circulations on $G$ as follows.

\begin{definition}
A flow $f \in \RR^{E(G)}$ is a \emph{circulation} if it
conserves flow on each vertex. In particular, for each $v \in V$ with
incoming edges $\Gamma^-(v)$ and outgoing edges $\Gamma^+(v)$, we have:

\[ \sum_{e \in \Gamma^-(v)} f(e)  = \sum_{e \in \Gamma^+(v)} f(e) \]
\end{definition}

We define the \emph{weight} of circulation $f$ as follows: $w(f) = \sum_{e
  \in E(G)} w_e f(e)$. Note that we allow a circulation to send
negative flow on an edge, as well as overflow the capacity of an
edge. A \emph{feasible} circulation is better behaved.
\begin{definition}
  A vector $f \in \RR^{E(G)}$ is a \emph{feasible circulation} if it is a
  circulation, and moreover it satisfies nonnegativity and
  capacity constraints. In particular, for each edge $e$ we have: \[ 0
  \leq f(e) \leq c_e\]
\end{definition}

We distinguish the simplest circulations, or \emph{flow cycles},
and recall that every circulation can be decomposed into cycles
that are oriented consistently.

\begin{definition}[Flow Cycle]
A circulation $C$ on $G$ is a \emph{flow cycle} if there exists
a simple undirected cycle $H$ such that $C$ sends non-zero flow only on edges
of $H$.
\end{definition}

\begin{definition}[Conformal Decomposition]
  Fix a circulation $f$. We say a set of circulations
  $\set{g^1,\ldots,g^k}$ is a \emph{conformal decomposition} of $f$ if the
  following hold
  \begin{enumerate}
  \item (Decomposition) $f = \sum_{i=1}^k g^i$ 
  \item (Conformal) For each $e \in E(G)$, $f(e) \cdot g^i(e) \geq 0$.
  \end{enumerate}

\end{definition}
Note that every $g^i$ in the conformal decomposition must send flow in
the same direction as $f$ on each edge.

\begin{theorem} [\cite{bertsekas}]
  \label{conformal-decomposition-cycles}
  Every circulation $f$ can be conformally decomposed
  into flow cycles.
\end{theorem}

Conformal decomposition immediately yields a useful property of feasible
circulations.

\begin{lemma}\label{lem:feasibility_maintained}
  Let $f$ and $f'$ be feasible circulations on $G$. Let
  $C_1,\ldots,C_k$ be a conformal decomposition into cycles of
  $f'-f$. Then, for every $L \sse \set{1,\ldots,k}$ the circulation $f +
  \sum_{i \in L} C_i$ is feasible. Equivalently, for every $L \sse
  \set{1,\ldots,k}$ the circulation $f' - \sum_{i \in L} C_i$ is feasible.
\end{lemma}
In other words, the cycles of the conformal decomposition of $f'-f$
may be added to $f$ (or subtracted from $f'$) in any order, maintaining feasibility along the way.






\end{document}